\documentclass[journal, 10pt]{IEEEtran}
\usepackage{amsmath, amsthm}
\usepackage{graphics,graphicx} 
\usepackage{array}
\usepackage{mdwmath}
\usepackage{amssymb}
\usepackage{mdwtab}
\usepackage[tight,footnotesize]{subfigure}
\usepackage[justification=centering]{caption}
\usepackage{algorithmic, algorithm}
\newtheorem{lemma}{Lemma}
\newtheorem{proposition}{Proposition}

\ifCLASSINFOpdf
\else
\fi

\begin{document}
\title{Targeting Customers  for Demand Response Based on Big Data}
\author{Jungsuk Kwac,~\IEEEmembership{Student Member,~IEEE} and Ram Rajagopal,~\IEEEmembership{Member,~IEEE}
\thanks{J. Kwac is with the Department of Electrical Engineering and the Stanford Sustainable Systems Lab, Department of Civil and Environmental Engineering, Stanford University, CA, 94305. Email: kwjusu1@stanford.edu.}%
\thanks{R. Rajagopal is with the  Stanford Sustainable Systems Lab, Department of Civil and Environmental Engineering, Stanford University, CA, 94305. R. Rajagopal is supported by the Powell Foundation Fellowship. Email: ramr@stanford.edu.} }


\maketitle
\begin{abstract}
Selecting customers for demand response programs is challenging and existing methodologies are hard to scale and poor in performance. 
The existing methods were limited by lack of temporal consumption information  at the individual customer level. 
We propose a scalable methodology for demand response targeting utilizing novel data available from smart meters. 
The approach relies on formulating the problem as a stochastic integer program involving predicted customer responses.
A novel approximation is developed  algorithm so it can scale to problems involving millions of customers. 
 The methodology is tested experimentally using real utility data. 
\end{abstract}

\begin{IEEEkeywords}
smart meter data, targeting, demand response, big data, algorithms, stochastic knapsack problem.
\end{IEEEkeywords}

\IEEEpeerreviewmaketitle
\section{Introduction}

Demand response (DR) programs have been expanded to match peak load growth and increases in supply-side uncertainty. 
The goal of a DR program is to elicit flexibility from loads by reducing or shifting consumption in response to external signals such as prices or curtailment indicators. 
Typically a program is designed to extract a targeted level of energy or power  from all the participating loads. 
The program \emph{operation yield} is the ratio between the actually extracted energy from the participants and the target level. 
Current program yields are low, in the range of 10\% to 30\%. 

Customer recruitment is essential for the success of demand response programs. 
Existing customer recruitment mechanisms are designed to identify customers that are likely to enroll. 
They utilize predictive marketing models (e.g. discrete choice) that utilize as inputs household related variables such as family income, enrollment in other programs and household size. The models are very successful in identifying customers that are likely to enroll, and typically 80\% of recruited customers enroll. Yet, the lack of individual level consumption information has prevented these recruitment efforts from achieving high operation yields. Besides yield, a DR program needs to ensure reliable performance. Customer enrollment directly influences reliability as each consumer might provide different levels of certainty on the demand curtailment they can offer during a DR event. 

The widespread deployment of  advanced metering infrastructure (AMI) data can significantly change the approach to customer recruitment. This paper investigates how to target customers for DR programs based on this data. The significant recruitment cost and scalable progressive deployment of DR programs require an efficient way to select a small number of appropriate customers from a large population.
Since enrollment decisions are made in advance of the actual consumption period, only a prediction of per customer DR potential is available. The prediction can be made by analyzing historical high resolution consumption data for each customer. Given a prediction, customers need to be chosen in a way that balances the magnitude of demand response potential (reward) with the uncertainty in the prediction (risk). In fact, customers need to be considered as a portfolio, and an optimal trade-off between risk and reward is desirable. 

This paper develops a methodology for large-scale targeting that combines data analytics and a scalable selection procedure. The first  contribution is the formulation of the customer selection problem as a stochastic discrete optimization program. This program selects the best customers in order to maximize reliability (risk)  given a targeted DR amount.  The second contribution is a novel guaranteed approximation algorithm for this class of discrete optimization problems which might be of independent interest. The methodology is illustrated in two climate zones of Pacific Gas \& Electric, with more than 50,000 customers.

The literature on demand response programs is vast (e.g. \cite{han2008solutions}, \cite{borenstein2002dynamic}, \cite{albadi2007demand}) and includes the definitions of DR programs \cite{rahimi2010demand}, how to estimate the energy savings (baselining) \cite{mathieu2011quantifying} and other related areas. Recent papers have focused on designs for operating DR including distributed pricing mechanisms that adapt to data (\cite{li2011optimal}, \cite{dong2012distributed}, \cite{chen2012will}, \cite{dominguez2011distributed}), 
 simplified operation mechanisms with full information (\cite{zhong2013coupon}, \cite{mathieu2011quantifying}) and operations with partial information \cite{taylor2012price}. Integration of DR as a system reliability has also been well investigated in \cite{xie2011wind}. Finally, implementation in hardware of these solutions is being developed in \cite{lee2013energy}. 
Currently, most demand response targeting relies on segmentation of customers based on their monthly billing data or surveys \cite{lutzen2009}, \cite{moss2008market}. The growing availability of smart meter data has shown that such approaches are highly inaccurate since segments obtained from actual consumption  
differ from those obtained by alternative methods \cite{kwac2014seg}, \cite{brian2012}.

The paper is organized as follows. Section~\ref{probstate} provides the problem description including a stochastic knapsack problem (SKP) setting and response modeling. Section~\ref{method} presents a review of SKP and develops a novel heuristic algorithm to solve the proposed SKP. Section~\ref{exp} presents the experimental evaluation and validation of the methodology. Section~\ref{con} summarizes the main conclusions and discusses future work.

\section{Methodology}\label{probstate}

Given the large amount of data, the scalability of the approach is very important. Thus, we propose a quick linear response modeling and a novel heuristic to solve the SKP \eqref{skp}, which is basically an NP-hard problem. Fig. \ref{oflow} shows the overall DR program targeting flow in this paper with brief computation complexity information.
\begin{figure}[htbp]
\includegraphics[width=3.5in,height=1.4in]{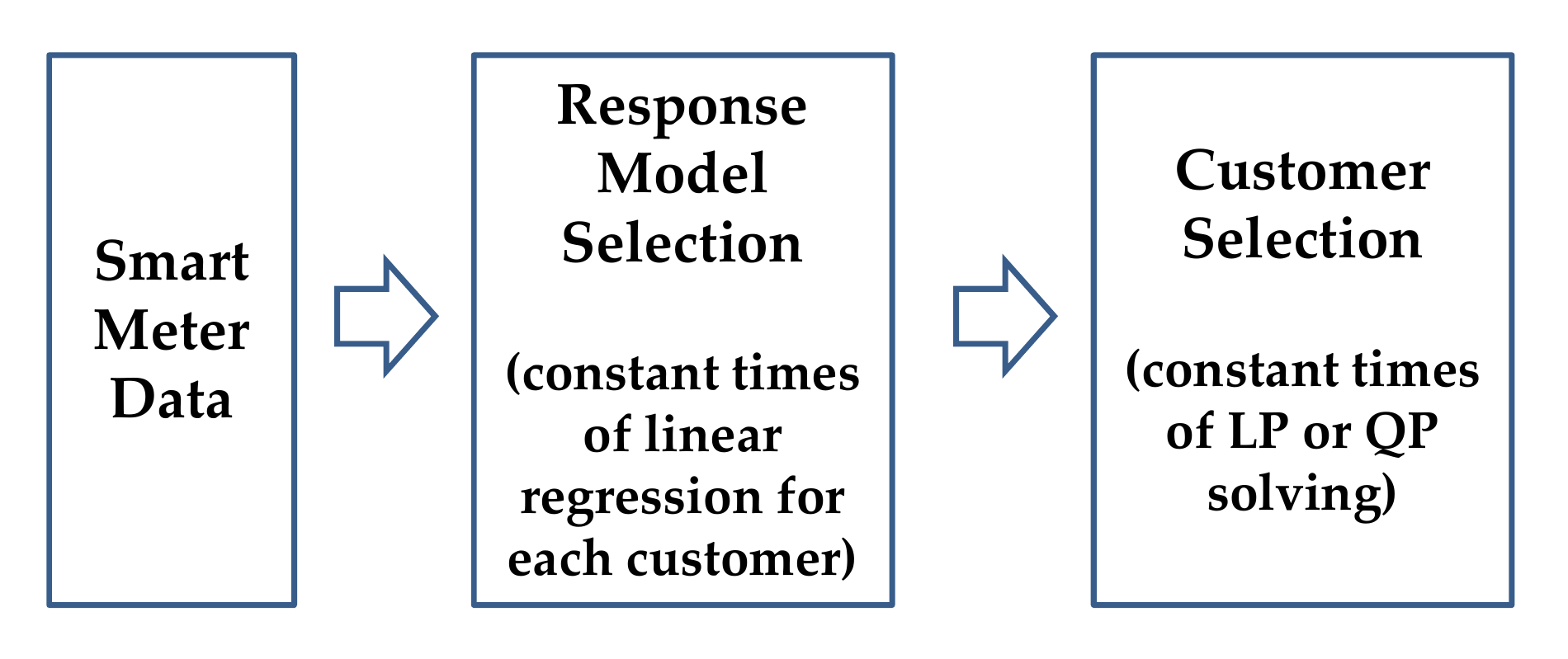}
\caption{Overall DR program targeting flow}\label{oflow}
\end{figure}

The remainder of the section details the methodology. It first proposes a reliability constrained selection algorithm that uses a probabilistic estimate of DR response. It then demonstrates how such response can be learnt from existing smart meter data. Notice that the specific approach to predict individual responses for different types of DR would differ, but each still  provides a \emph{probabilistic} estimate. For completeness, we illustrate the procedure for thermal DR. 

\subsection{Maximizing Demand Response Reliability}


There are $K$ potential customers that can provide DR service. Given the customer recorded data, the response of customer $k$ is a random variable $r_k$ corresponding to the energy saved during a DR event.  The distribution of $r_k$ is determined by fitting a response model corresponding to the type of DR and has a known joint probability distribution. The cost for customer $k$ to participate in the program is $c_k$. During planning, this cost represents the cost of marketing the program to a customer and rebates for a customer to purchase the resources to perform DR. The program operator has a budget  $C$ and desires an aggregate target response $T$ (in kWh) from the program with the maximum reliability possible. \emph{DR availability} is captured by the control variable $T$.  \emph{DR reliability} is naturally captured by the probability of the response exceeding the target $T$. The optimal DR program selection problem can then be stated  as
\begin{align} \label{skp}
&\max_{\bold{x}} \;P\left(\sum_{k=1}^{K} r_{k}x_{k} \ge T\right),\\
&\textrm{s.t.} \;\;\; \sum_{k=1}^{K} c_{k}x_{k} \le C, \;\;(\Leftrightarrow \sum_{k=1}^{K} x_{k} \le N, \textrm{ if $c_k$ is same} ),\nonumber\\
&\;\;\;\;\;\;\;\;\;x_{k} \in \{0,1\}\;\;, k = 1,...,K\nonumber
\end{align}
where $x_k$ represents whether a customer is recruited or not. Note that, if $c_k$ is same for all customers, the budget constraint is the same as limiting the number of participating customers by $N$. The program maximizes the reliability  of the DR program by recruiting customers within the program budget. The optimal reliability for  budget $C$ and target $T$ is given by objective function value  $p^*(C,T)$.  The function captures the tradeoff between DR availability and DR reliability for a budget $C$. The function has some important properties that conform to our intuition about the tradeoff. The objective function is monotonic decreasing in $T$ so $p^*(C,T_1) \geq p^*(C,T_2) $ if $T_1 \geq T_2$. The budget determines the constraints so $p^*(C_1,T) \leq p^*(C_2,T) $ if $C_1 \geq C_2$. 

The proposed optimization problem is a stochastic knapsack problem (SKP). SKPs are stochastic integer programs known to be NP-hard. The goal of this paper is to develop a novel efficient approximation algorithm that scales to $K$ in millions of customers. The efficient algorithm is used then to compute the function $p^*(C,T)$.  

An important additional assumption is that $K$ is large and $C$ is sufficiently large, so a significant number of customers are included. In that case, given a set of random variables for the response $r_k$, the total group response is approximately Gaussian distributed due to the central limit theorem: 
\begin{align} \label{skp2}
\sum_{k=1}^{K} r_{k}x_{k} &\sim N(\boldsymbol{\mu}^T \bold{x},\bold{x}^T\bold{\Sigma x}),
\end{align}
where $\;\bold{x}$ is the recruitment vector,  $N$ indicates a normal distribution, $\boldsymbol{\mu}$ is a vector of individual response means $\mu_k$ and $\bold{\Sigma}$ a covariance matrix with covariances $\Sigma_{jk}$ between responses. In practice, if the number of customers selected is in the order of 50, the response distribution is very close to normal. 

\subsection{Response modeling}

Estimating the response $r_k$ from provided data is an important challenge in practical applications. We illustrate here a model for estimating customer response in a specific program, but we note that our methodology applies in general, with the ability to define general models. A highlight of the methodology is that $r_k$ is a random variable, so models that are not accurate can still be useful. 

The customer response model specification depends on the design of the demand response program. Consider a Global Temperature Adjustment (GTA) program for HVAC (air conditioning) systems \cite{han2008solutions}. Such a program increases the temperature set point of the air conditioner for each customer by a fixed amount to reduce cooling power consumption.  Selecting customers with high energy saving potential during a DR event day and hour requires an accurate model to estimate the total energy consumed at each set point level.  If HVAC consumption is independently metered, a simple model can be built utilizing the observed consumption, external temperature and the utilized set point \cite{sonderegger1978dynamic}. In general, though, only total home consumption and external temperature are observed. 

The main proportion of temperature-sensitive energy consumption is from HVAC (heating, ventilation, and air conditioning) systems \cite{han2008solutions}. It has been observed that the relationship between energy consumption and outside temperature is piecewise linear, with two or more breakpoints \cite{mathieu2011quantifying}, \cite{birt2012disaggregating}, \cite{cancelo2008forecasting}. These descriptive models are typically fit utilizing past consumption data, in particular for the summer season. The power consumption of a customer $k$ at time $t$ on day $d$, $l_k(t,d)$, is modeled as 
\begin{align}\label{intermodel}
l_k(t,d) &= a_k(t)(To_k(t,d) - Tr_k)_+  \\
&+ b_k(t)(Tr_k - To_k(t,d))_+ + c_k(t) + e_k(t),\nonumber
\end{align}
where $To_k(t,d)$ is the outside temperature, $Tr_k$ is the breakpoint temperature which is the proxy of the temperature set point for the HVAC system, $a_k(t)$ is the cooling sensitivity, $b_k(t)$ is the heating sensitivity or the temperature sensitivity before turning on AC system, $c_k(t)$ is base load and $e_k(t)$ is a random variability. Typically, in the summer $b_k(t)$ is close to zero. In some cases,  $b_k(t) \approx a_k(t)$ since either the variability is large or only large temperatures are observed in the summer, thus a reliable estimate of $b_k(t)$ cannot be obtained. The model above is  closely related to the ETP model \cite{sonderegger1978dynamic}, assuming that the system is in thermal equilibrium so inner mass temperature and outside temperature are equal. 

To restrict the computational cost to find the best breakpoint, $Tr_k$ is assumed to be an integer, in the range 68-86${}^\circ F$, which is typical. Additionally, to prevent the cases that one or two dominant data points determine the breakpoint and provide invalid temperature sensitivity, we put a constraint on the breakpoint that both sides of the breakpoint should have at least certain fraction of data (e.g., we set 15\% in this paper). 

Model learning is performed in two steps. Minimization of residual sum of squares (RSS) is utilized to learn the parameters of the model $\{Tr_k, a_k(t), b_k(t), c_k(t)\}$ and the distribution of the error $e_k(t)$ from the observed data $l_k(t,d)$ and $To_k(t,d)$. An F-test is utilized to test if $a_k(t) = b_k(t)$ to prevent overfitting. If $a_k(t) = b_k(t)$, \eqref{intermodel} becomes \eqref{basic}, 
\begin{align}\label{basic}
l_k(t,d) &= a_k(t)To_k(t,d) + c'_k(t) + e_k(t).
\end{align}

The overall computation needed to fit the consumption model is to solve (at most) 20 linear regression models: one for each potential value of the breakpoint (at most 19) and one for fitting \eqref{basic}. The coefficients associated with the breakpoint with smallest RSS is selected as the estimate. The regression provides estimates of parameter values and errors, in particular an estimate of cooling sensitivity $\hat{a}_k(t)$ and the standard deviation of parameter error denoted by $\sigma (\hat{a}_k(t))$.  The distribution of the parameter estimate might not necessarily be Gaussian. Covariances $\textrm{COV}(\hat{a}_j(t), \hat{a}_k(t))$  between sensitivity estimates for different customers can be obtained by analyzing the residuals. 

The DR response model is postulated as follows. The savings in the response are obtained by increasing the set point temperature by $\Delta Tr_k$. We assume DR is utilized during very hot days, so $To_k(t,d) \geq Tr_k + \Delta Tr_k$. The response model is then postulated to be $r_k = {a}_k(t)\Delta Tr_k$ kWh.  The random variable $r_k$ has mean $\mu_k = \hat{a}_k(t)\Delta Tr_k$ and standard deviation 
$\sigma_k=\sqrt{\Sigma_{kk}}=\sigma(\hat{a}_k(t))\Delta Tr_k$. The covariance between any two responses $r_j$ and $r_{k}$ is given by $\Sigma_{jk} =  \textrm{COV}(\hat{a}_j(t), \hat{a}_k(t)) \Delta Tr_j \Delta Tr_k$.  The aggregate DR response is then distributed as in \eqref{skp2}.

\section{Algorithm}\label{method}
In this section we design a novel scalable algorithm for customer selection that enables construction of the DR availability-reliability tradeoff.  

\subsection{Optimization problem transformation}\label{trans}

Utilizing the normality assumption of the total response in the original SKP formulation \eqref{skp}, the following equivalence is easy to demonstrate:
\begin{align} \label{skp3}
p^*=\max_{\bold{x} \in \mathcal{C}} P\left(\sum_{k=1}^{K} r_{k}x_{k} \ge T\right) &\Leftrightarrow \; \rho^* = \min_{\bold{x} \in \mathcal{C}} \frac{T-\boldsymbol{\mu}^T \bold{x}}{\sqrt{\bold{x}^T\bold{\Sigma x}}}. 
\end{align}
The set $\mathcal{C}$ represents the constraint set of the original SKP. The equivalency reveals some qualitative insights about the optimal solution. Consider the case $c_k = 1$,  so the budget $C$ is the number of customers.If $C$ is small, then its likely the target level  $T$ exceeds the group response so $\rho^* \geq 0$. The mechanism mostly selects  customers with high mean, high variance  responses resulting in low reliability. If $C$ is sufficiently large, then $\rho^* < 0$, and the mechanism selects high mean, low variance customer responses. Consequently DR reliability is high. A simple rule of thumb is that $C$ should be usually set so that the sum of the $C$ highest response means $\mu_k$ exceeds $T$. 

\subsection{Previous approaches to solve the SKP}\label{prior}

Several approaches address solving the transformed SKP in \eqref{skp3} (\cite{geoffrion1967solving}, \cite{henig1990risk}). The basic principle is to utilize \textit{Lemma~\ref{lem-skp}} proved independently for completeness.  
\begin{lemma}\label{lem-skp}
The customer selection SKP~\eqref{skp} is equivalent to the optimization program 
\begin{align} \label{simple_skp}
&\min_{\mu,\sigma^2} \frac{T- \mu}{\sqrt{\sigma^2}} \\
&\textrm{s.t.} \;\;\; (\mu, \sigma^2) \in \mathcal{H} \nonumber,
\end{align} 
where the constraint set is defined as 
\begin{align*} 
 \mathcal{H} &=  \textrm{Conv}\{(\mu, \sigma^2): \bold{x} \in \mathcal{H}_X, \mu=\boldsymbol{\mu}^T \bold{x}, \sigma^2 = \bold{x}^T\bold{\Sigma} \bold{x}  \},\\
 \mathcal{H}_X &=\left\{\bold{x}: \bold{x} \in \{0,1\}^K,  \bold{c}^T\bold{x} \leq C, \boldsymbol{\mu}^T \bold{x}  \geq T,  \bold{x}^T\bold{\Sigma} \bold{x} \geq 0\right\},
\end{align*} 
under the assumption that $\mu^* \geq T$. Furthermore,  there is a map from the optimum $(\mu^*, \sigma^{*2})$ to $\bold{x}^*$. Additionally, there exists $0 \leq \lambda \leq 1$ so the optimization 
\begin{align}
&\max_{\mu,\sigma^2} \lambda \mu - (1-\lambda) \sigma^2  \label{obj2}\\
&\textrm{s.t.} \;\;\; (\mu, \sigma^2) \in \mathcal{H} \nonumber,
\end{align} 
has the same optimum $(\mu^*, \sigma^{*2})$ as optimization~\eqref{simple_skp}. 
\end{lemma} 
\begin{proof}
Define $\rho(\mu,\sigma^2) =  (T- \mu)/\sqrt{\sigma^2}$. 
Assuming there exists a $\bold{x}$ so that $\mu=\boldsymbol{\mu}^T \bold{x}$ and  $\sigma^2 = \bold{x}^T\bold{\Sigma} \bold{x}$, the function $\rho$ is the same as the equivalent optimization in \eqref{skp3}. The gradient  
$\vec\nabla \rho(\mu,\sigma^2) = (-1/\sqrt{\sigma^2}, -(T-\mu)/(2\sigma^2\sqrt{\sigma^2})$ shows $\rho(\mu,\sigma^2)$ 
monotonic decreasing in $\mu$ and increasing in $\sigma^2$. 
on the set $\{(\mu, \sigma^2): \mu > T, \sigma^2 > 0\}$.

By $Lemma\;1$ in \cite{henig1990risk}, $\rho(\mu,\sigma^2)$ is pseudoconcave on the set. 
Based on the fact that a quasiconcave function over a convex set attains its minimum in an extreme point of the constraint set \cite{bazaraa2013nonlinear}, 
the extreme points of $ \mathcal{H}$ 
should be found and tested to see which extreme point is the optimal solution. The extreme points in a convex set can be found by linear programming \cite{boyd2004convex}. As shown in the gradient $\vec\nabla \rho(\mu,\sigma^2)$ 
, when $\rho^* < 0$ ($\mu^* > T$, so $T$ and $C$ are appropriately selected), the objective function in \eqref{simple_skp} is minimized with larger $\mu$ and smaller $\sigma^2$. Thus, the extreme points which are feasible to be the optimal solution can be obtained by solving the problem below.
\begin{align*} 
\max_{\mu,\sigma^2} \lambda \mu - (1-\lambda) \sigma^2, (0 \leq \lambda \leq 1)
\end{align*} 

\end{proof} 

To find the extreme points in $\mathcal{H}$, \eqref{obj2} 
is solved multiple times for various $\lambda$. When $\lambda$ is given, the calculation can be done by dynamic programming (DP). However, the overall procedure will be determined mainly by selecting a small number of various values of $\lambda$ to find all extreme points. Setting various $\lambda$ is related to nonlinear optimization techniques referenced in \cite{henig1990risk}. According to \cite{henig1986shortest}, the expected complexity of finding the set of various $\lambda$ is $O({_K}C_{N})$ (where $K$ is the number of customers, $N$ is the limit of the number of participating customers), which requires heavy computation cost when $K$ is large.

Alternatively,  \cite{morton1998stochastic} does not try to find extreme points in $\mathcal{H}$, but rather changes the SKP into multiple knapsack problems (KP) on the assumption that all the $\sigma^2$s are integers and all users are independent 
 ($\bold{\Sigma}$ is diagonal). The method is quite computationally intensive and requires solving $O(N\lfloor \max(\sigma^2)\rfloor)$ integer LP. 
 
A more straightforward approach is to extend the concept of a greedy approximation algorithm utilized to solve specific instances of the knapsack problem  \cite{dantzig1957discrete, dean2004approximating}. The proposed algorithm is described in Appendix~\ref{sec:appb} as Algorithm~\ref{alga3} and has complexity $O(K\log K)$.  It utilizes the  per customer risk-reward ratio $\mu_k/\sigma_k$ to sort and rank customers who are offering sufficient benefit. This constraint improves the performance of the algorithm as explained in the Appendix.

\subsection{Stochastic Knapsack problem solving}\label{optsol}

We propose an efficient heuristic algorithm (Algorithm \ref{alga2}) to solve the SKP by finding the extreme points in $\mathcal{H}$. 
By the monotonicity property shown in the proof of Lemma~\ref{lem-skp}, the maximum should be obtained from one of the extreme points found by solving \eqref{obj2} with various $\lambda$. 
Every extreme point of $\mathcal{H}$ has a corresponding point $\bold{x} \in \mathcal{H}_X$ which can be obtained by solving
\begin{align} \label{subprob}
\max_{\bold{x} \in \mathcal{H}_X} \lambda' \boldsymbol{\mu}^T \bold{x} - \bold{x}^T\bold{\Sigma x}, \;(\lambda' = \frac{\lambda}{1-\lambda}\ge0).
\end{align} 

If we think a two-dimension scatter plot of $\mathcal{H}$ (x-axis: $\mu$, y-axis: $\sigma^2$), $\lambda'$ corresponds to the slope. Then, finding an extreme point by solving \eqref{subprob} corresponds to finding the point on the slope with minimum intercept in $\mathcal{H}$. 
Depending on the slope, $\lambda'$, different extreme points will be obtained. Thus, the heuristic we propose here is to find all extreme points by solving \eqref{subprob} 
 constant times, $M$, i.e., increasing the slope from $0$ to $\pi/2$ by $\pi/2M$. 

\begin{algorithm}
\caption{Algorithm to solve the SKP in \eqref{skp}}\label{alga2}
\begin{algorithmic}
\REQUIRE $\bold{\mu}$ and $\bold{\Sigma}$ from response modeling.
\STATE Set a integer constant, $M$ (= the number of iterations).
\FOR {$i$ from 0 to $M$}
\STATE $\lambda_i' = \tan\left(\frac{i\pi}{2M}\right)$ (= equally increasing slope angle).
\STATE \; Solve the problem below and save $\bold{x_i}$ : 
\begin{align}
\textrm{(If $\rho^*\le 0$)}\; \bold{x_i} = \arg &\max_{\bold{x}} \lambda_i' \boldsymbol{\mu}^T \bold{x} - \bold{x}^T\bold{\Sigma x},\label{findextreme}\\
\textrm{(If $\rho^*> 0$)}\; \bold{x_i} = \arg &\max_{\bold{x}} \lambda_i' \boldsymbol{\mu}^T \bold{x} + \bold{x}^T\bold{\Sigma x},\label{findextreme2}\\
&s.t.\;\;\; \sum_{k=1}^{K} c_{k}x_{k} \le C,\nonumber\\
&\;\;\;\;\;\;\;\;\;x_{k} \in \{0,1\} \;\;  \forall k.\nonumber
\end{align}
\ENDFOR
\begin{equation}
\text{Return} \;\bold{\bar{x}} = \arg \min_{\bold{x}\in\{\bold{x_0},...,\bold{x_M}\}} \frac{T-\boldsymbol{\mu}^T \bold{x}}{\sqrt{\bold{x}^T\bold{\Sigma x}}}.
\hspace{1in}\nonumber
\end{equation}
\end{algorithmic}
\end{algorithm}


We bound the ratio between the optimal cost obtained by Algorithm~\ref{alga2} and the true optimal cost $\rho^*$ as defined in \eqref{skp3}: 

\begin{proposition}\label{prop1}
Let $\mu'_i = \boldsymbol{\mu}^T \bold{x_i}, \sigma'^2_i = \bold{x_i}^T\bold{\Sigma x_i}$. When $\rho^*$ in \eqref{skp3} is less than zero, Algorithm \ref{alga2} has the approximation bound below, which is only depending on $\sigma'_i$.
\begin{align}
\frac{T-\boldsymbol{\mu}^T \bold{\bar{x}}}{\sqrt{\bold{\bar{x}}^T\bold{\Sigma \bar{x}}}} < \rho^* \min_i \frac{\sigma'_{i-1}}{\sigma'_i}.
\end{align}
\end{proposition}

We provide the proof for \textit{Proposition \ref{prop1}} in Appendix~\ref{approxproof}. 
Briefly, the bound depends on the given data and the number of iterations, $M$, which decide $\sigma'_i$. For example, we provide the relation between $M$ and the approximation bound at Zone 13 in section \ref{subsec:selection} by Fig. \ref{mvsapprox}. 


When $\rho^*$ in \eqref{skp3} is larger than zero, under a given $T$ and $N$, the optimal solution does not have to be one of the extreme points \cite{henig1990risk}. However, we assume that the optimal solution stays at one of the extreme points and find the extreme points by solving \eqref{findextreme2} 
instead of \eqref{findextreme} in Algorithm \ref{alga2}. This is because the optimization direction is increasing $\boldsymbol{\mu}^T \bold{x}$ and $\bold{x}^T\bold{\Sigma x}$ as much as possible when $\rho^*$ is larger than zero.

When $\bold{\Sigma}$ is not diagonal, \eqref{findextreme} should be solved by quadratic programming (QP) after relaxing $0\le x_k\le1$. However, if we assume that the responses are independent so that $\bold{\Sigma}$ is diagonal ($\boldsymbol{\sigma}^2 = (\Sigma_{11}, ..., \Sigma_{KK}))$, solving the problem \eqref{findextreme} becomes a linear programming (LP) problem as shown in \eqref{indepcase}. 
\begin{align}\label{indepcase}
\lambda'_i \boldsymbol{\mu}^T \bold{x} - \bold{x}^T\bold{\Sigma x} &= (\lambda'_i \boldsymbol{\mu} - \boldsymbol{\sigma}^2)^T\bold{x}.
\end{align}
Especially, when the cost for each customer, $c_k$, is the same, the problem simplifies to selecting the highest (at most) $N$ entries from the $\lambda'_i \boldsymbol{\mu} - \boldsymbol{\sigma}^2$ vector, which is the same as sorting the $K$ entries, $O(K log K)$, in a computational perspective. 

To summarize, with the assumption of having independent responses and the same targeting costs, our customer selection procedure guarantees a very close to optimal solution, with a computational complexity ($O(K log K)$) equivalent to that of sorting $K$ entries in a vector. 
This is the most significant benefit of our heuristic algorithm, which enables customer selection even with a very large number of customers.

\section{Experiments on data}\label{exp}
\subsection{Description of data}
The \emph{anonymized} smart meter data used in this paper is provided by Pacific Gas and Electric Company (PG\&E). The data contains the electricity consumption of residential PG\&E customers at 1 hour intervals. There are 218,090 smart meters, and the total number of 24 hour load profiles is 66,434,179. The data corresponds to 520 different zip codes, and covers climate zones 1, 2, 3, 4, 11, 12, 13 and 16 according to the California Energy Commission (CEC) climate zone definition. The targeting methodology is tested by focusing on a cool climate zone (Zone 3: 32,339 customers) and a hot climate zone (Zone 13: 25,954 customers) during the summer season May to July 2011. Zone 3 is a coastal area and Zone 13 an inland area. Weather data corresponding to the same period is obtained from Weather Underground for each zip code in these two zones.

\subsection{Consumption model fitting result}
In this section, we provide analysis on the consumption model fitting result. First, the model selection result (F-test) is shown in order to check 
which model (between \eqref{intermodel} and \eqref{basic}) is selected for the customers depending on their climate zones. 
Moreover, to check whether the model explains the consumption effectively, we provide $R^2$ value distribution for each climate zone. As noted before, the models are aimed to capture the energy consumption of AC systems during the summer season. 25,954 household data in Zone 13 were used to fit the model, and, for comparison purposes, another 32,339 household data in Zone 3 were used.
\begin{figure}[htbp]
\includegraphics[width=3.3in,height=2.1in]{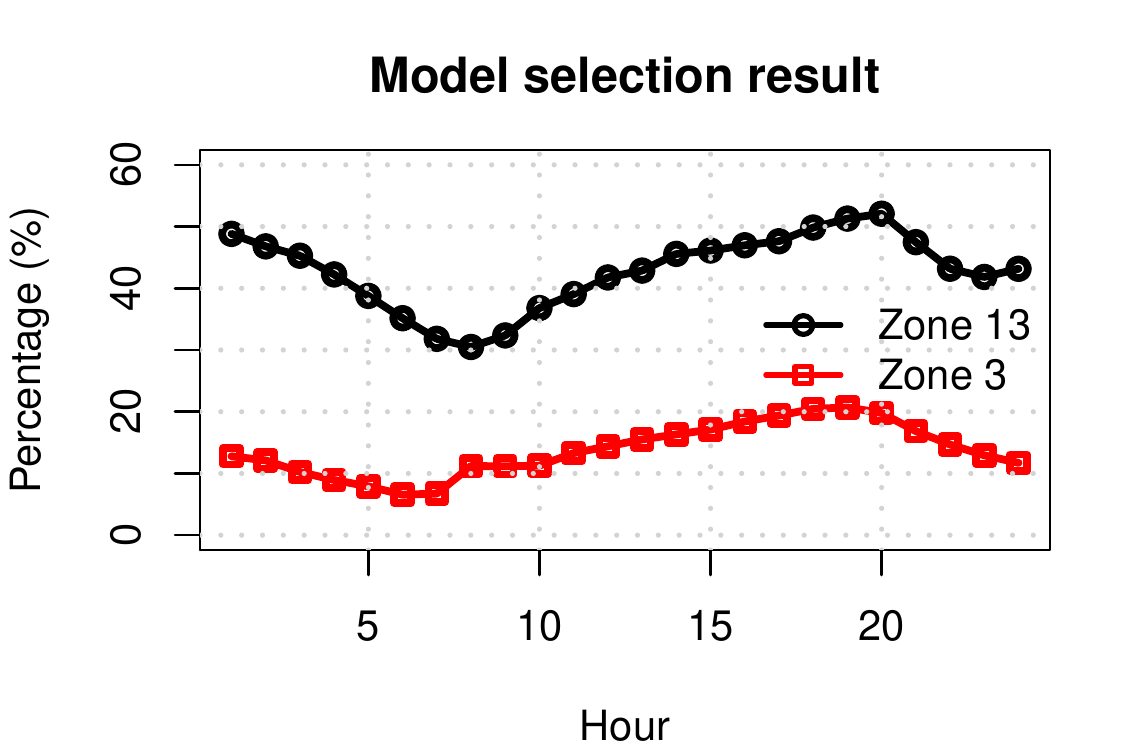}
\caption{Hourly model selection result by F-test \\(the percentage of that model \eqref{intermodel} is selected)
}\label{hourlyftest}
\end{figure}

Note that the degree of freedom of model \eqref{intermodel} is overrated because we constrained on the reference temperature to be an integer value between 68 and 86. Therefore, the F-test result will show a conservative view of the performance of the model \eqref{intermodel}. Fig. \ref{hourlyftest} shows the hourly model selection result for both climate zones. It shows the percentage of times the F-test selects model \eqref{intermodel} 
for each hour. As shown in the figure, 
model \eqref{intermodel} is selected 
much more in Zone 13 (hot) than Zone 3 (cool), which means model \eqref{intermodel} does not perform significantly better than \eqref{basic} in the cool climate zone. Also, we note that model \eqref{intermodel} fits better in higher consumption hours (3-9PM) for both climate zones.

Considering that the system peak hours of consumption for a typical utility occur between 4PM and 6PM during the summer season, we assume that these are the hours when DR is most required. Thus, we provide the result on those hours. 
Fig. \ref{r2both} shows the $R^2$ distribution during 5PM to 6PM. The other hours show almost similar results with that of 5PM to 6PM. In the hot area, the model explains about 42\% of variances on average, while it only explains 10\% of variances in the cool area. Considering that the only independent variable is temperature, temperature should be a very important factor in the summer season in the hot climate area, while it is not important factor of energy consumption in the cool climate area. The result is expected since cool areas tend to utilize air conditioning less frequently or even avoid it all together. 
\begin{figure}[htbp]
\includegraphics[width=3.5in,height=1.2in]{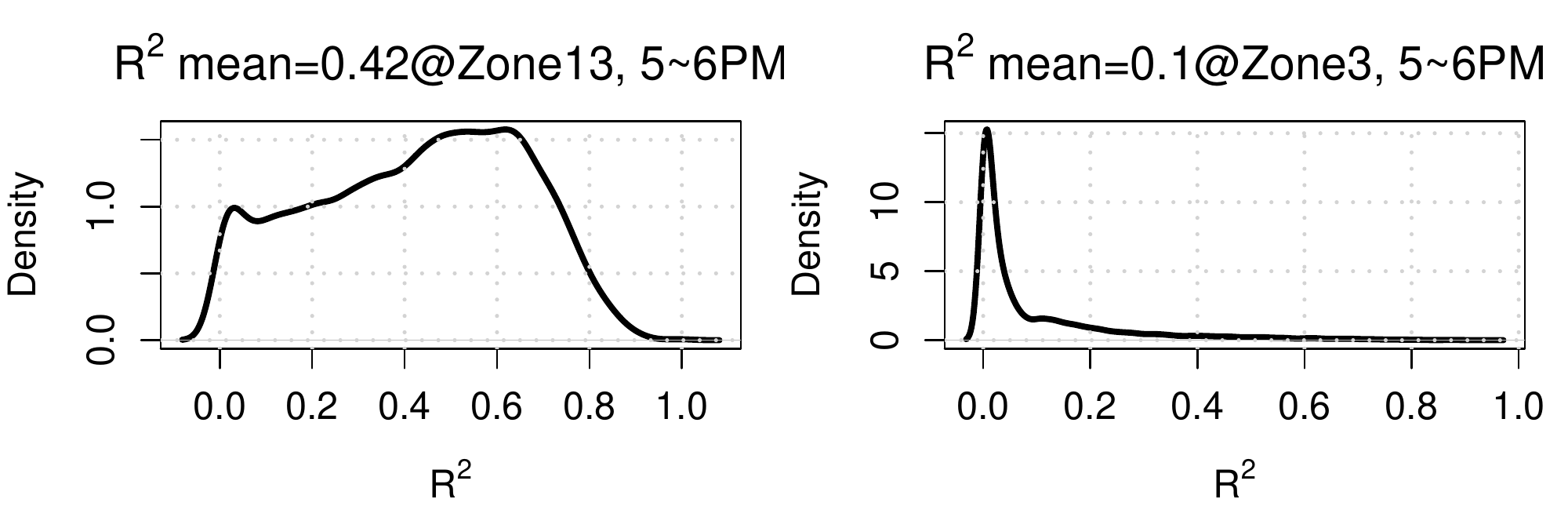}
\caption{$R^2$ distribution at both climate zones}\label{r2both}
\end{figure}

Fig. \ref{musigma} shows the smoothed scatter plots of $\mu_k$ (mean response) and $\sigma_k$ (response standard deviation) in \eqref{skp2} in both climate zones when 
$\Delta Tr_k$ is $3{}^\circ F$. Most of customers in Zone 13 have positive mean responses (coming from positive temperature sensitivity) with relatively small standard deviations, while most of the customers in Zone 3 have close to zero response, and only small fraction of customers have positive mean responses with relatively widespread standard deviations. This suggests that the explained DR program would not be as effective in the cool area as in the hot area. This is why we can achieve the same energy saving with a smaller $N$ in the hot area, as is shown in the following section \ref{subsec:selection}.
\begin{figure}[htbp]
\includegraphics[width=3.5in,height=2.1in]{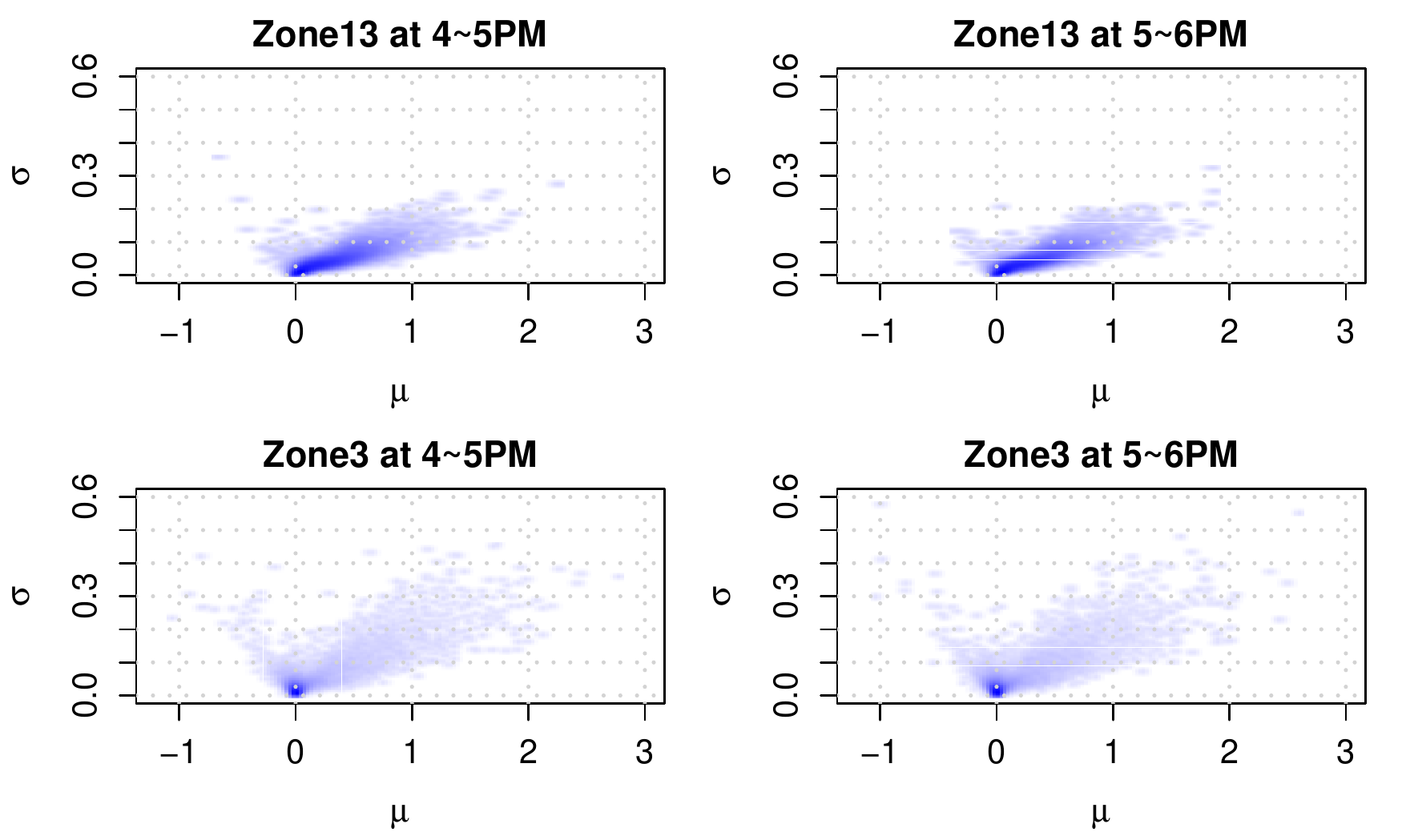}
\caption{$\mu$ and $\sigma$ scatter plot in both climate zones}\label{musigma}
\end{figure}


\subsection{Targeting result analysis}\label{subsec:selection}
Section~\ref{probstate} developed a SKP for optimal customer selection, and two different algorithms are provided to solve this combinatorial problem in Section~\ref{optsol}: an efficient heuristic (Algorithm \ref{alga2}) and a greedy algorithm (Algorithm \ref{alga3}). In this section, we show the customer selection results by solving the SKP \eqref{skp} with the assumptions that $c_k$ is same for all $k$ and $\bold{\Sigma}$ is diagonal.

First, we provide the tradeoff curve between DR availability and DR reliability with fixed $N$, which is the most important statistics for a utility manager with a limited budget to run a DR program. Then, we present the relation between varying $T$ (availability) and the minimum $N$ that achieves 95\% probability. Through this plot, we can provide how much cost is required to achieve certain energy saving. Also, we show the relation between varying $N$ and the maximum probability (reliability) in \eqref{skp} given $T$, which is corresponding to the tradeoff between DR reliability and the cost with a given $T$.  Last, we provide the plot of $M$ and the approximation bound.

\begin{figure}[htbp]
\includegraphics[width=3.5in,height=2.5in]{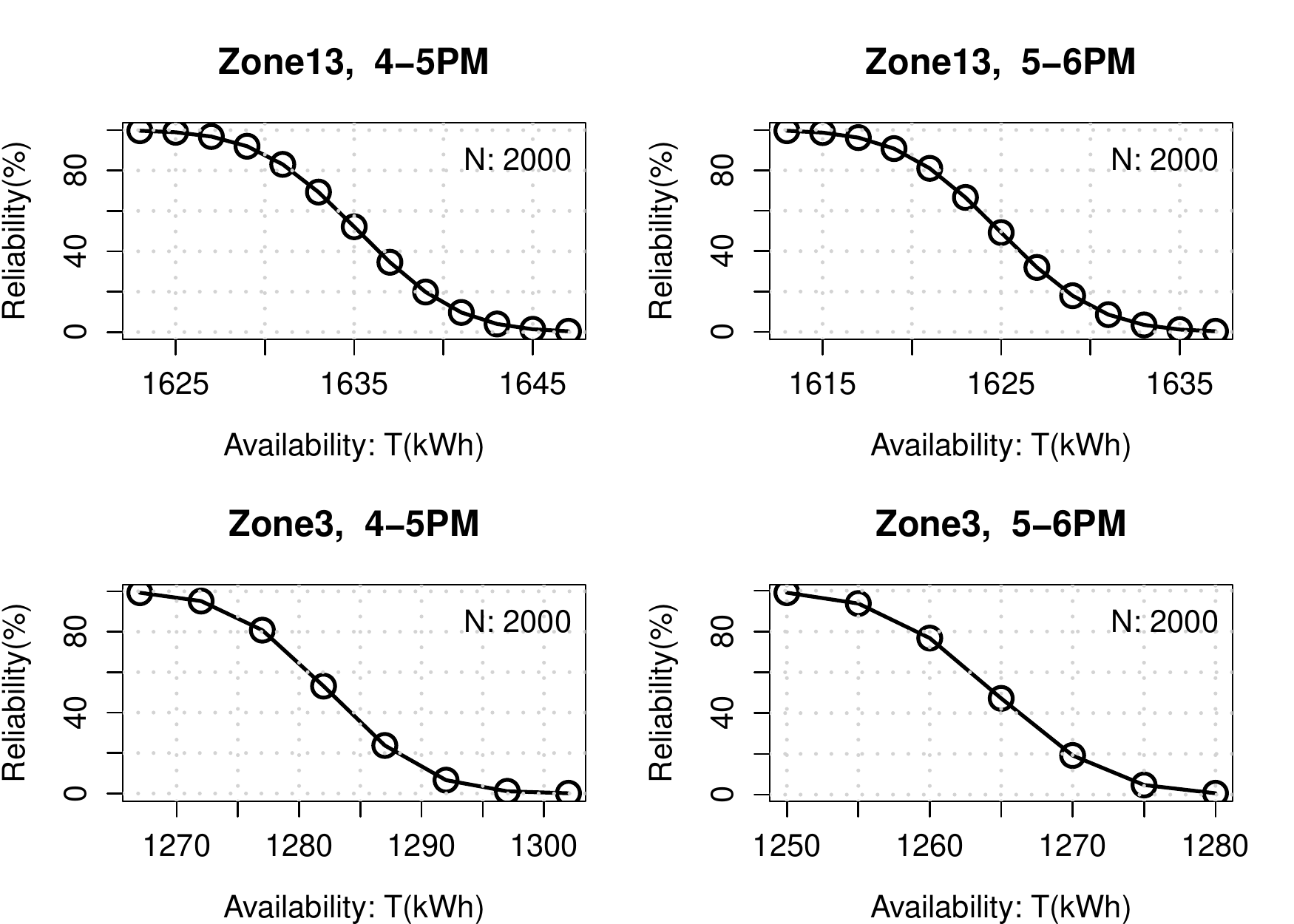}
\caption{Tradeoff curve between DR availability and reliability}\label{relava_both}
\end{figure}
Fig. \ref{relava_both} shows the DR availability-reliability tradeoff curve within a certain limit of the number of participating customers. The interpretation example can be that, if a utility wants to save 1625kWh during 5-6PM at Zone 13, we can not guarantee the saving more than 50\% with even with the best 2000 customers. 
A utility manager can generate this plot with setting $N$ depending on their budget and decide how much energy they will target to save with a certain reliability sacrifice.

\begin{figure}[htbp]
\includegraphics[width=3.5in,height=2.5in]{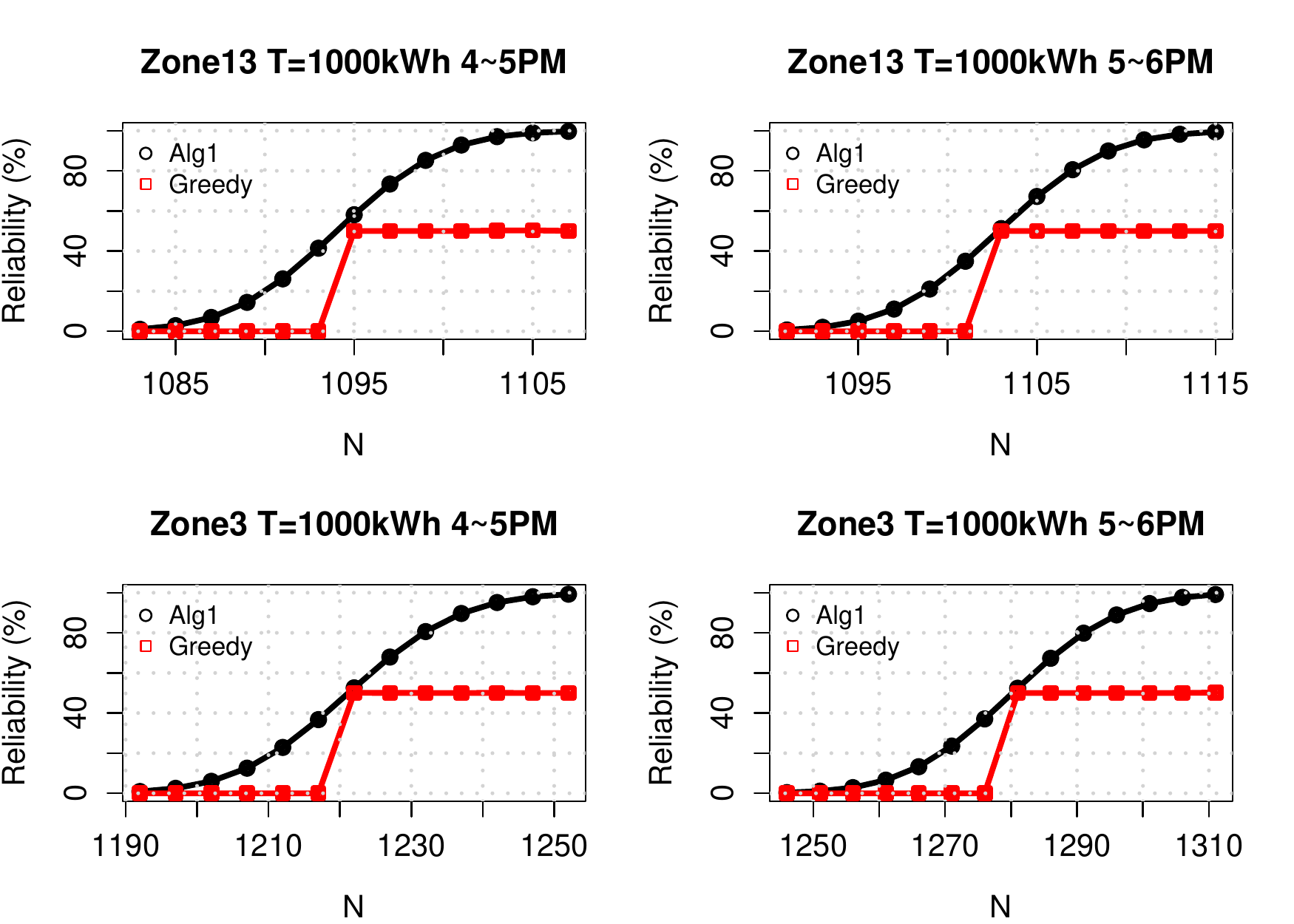}
\caption{Maximum probability (reliability) and N for a DR program with $\Delta Tr_k$ = $3{}^\circ F$ }\label{pmaxnbothindep}
\end{figure}

Fig. \ref{pmaxnbothindep} 
shows the relation between maximum probability and $N$ given $T$ in \eqref{skp}. We assume that DR events can happen for each hour during 4-6PM, $\Delta Tr_k$ is $3{}^\circ F$ and $T$ is 1000kWh. According to the plots, the heuristic algorithm always achieves equal or higher probability than the greedy algorithm. As designed, the greedy algorithm guarantees more than 50\% probability when the optimal solution can achieve more than 50\%. One thing to note is that the reliability changes from 0\% to 100\% within 30 customers in the hot climate zone though about 1100 customers are required to achieve 1000kWh of energy saving. This implies that an inappropriate number of enrollments can cause a program to fail, reinforcing the need for the analytic targeting mechanism. In contrast, in the cool climate zone, it takes more than 60 customers to reach a reliability from 0\% to 100\%. As shown in Fig. \ref{musigma}, the reason is that a small number of the customers with positive mean responses in Zone 3 have relatively large standard deviations compared to the customers in Zone 13. 

\begin{figure}[htbp]
\includegraphics[width=3.5in,height=2.5in]{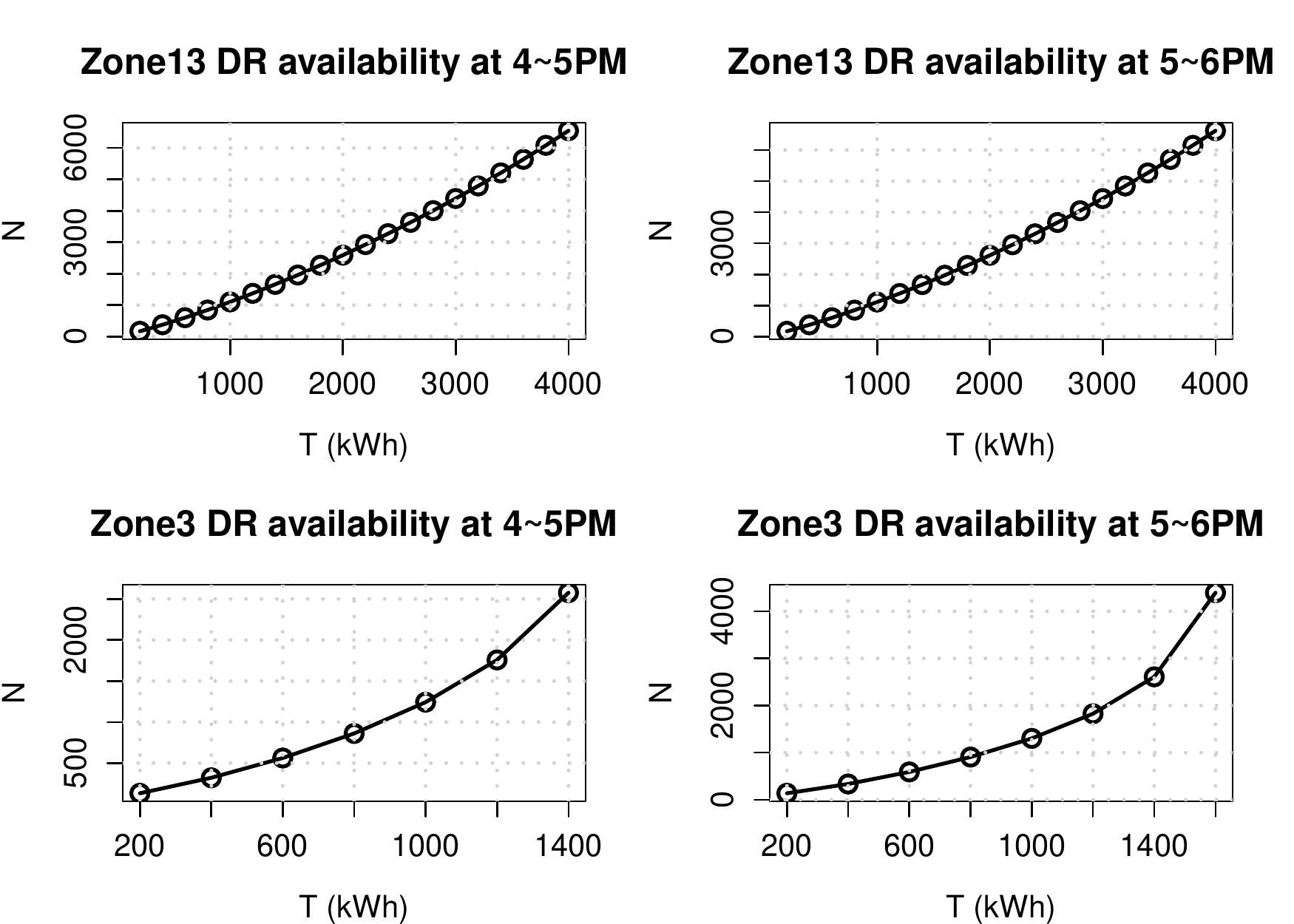}
\caption{$T$ and the minimum $N$ to achieve $T$ with more than 95\% probability}\label{tnmin_both}
\end{figure}
Fig. \ref{tnmin_both} 
shows the relation between varying $T$ (availability) and the minimum $N$ required to achieve the corresponding $T$ with guaranteeing the reliability ($>$95\%). We varied T from 200kWh to 4000kWh for both climate zones. Note that energy savings of 4000kWh or more are achievable with high probability in Zone 13 while it is not possible to achieve even more than 1600kWh in Zone 3. This fact supports the reasonable conjecture again that targeting customers in hot climate zones is more effective in achieving energy savings than targeting customers in cool climate zones. Also, note that the relation between $T$ and $N$ in the figures can be approximated by quadratic equations.

\begin{figure}[htbp]
\includegraphics[width=3.5in,height=2in]{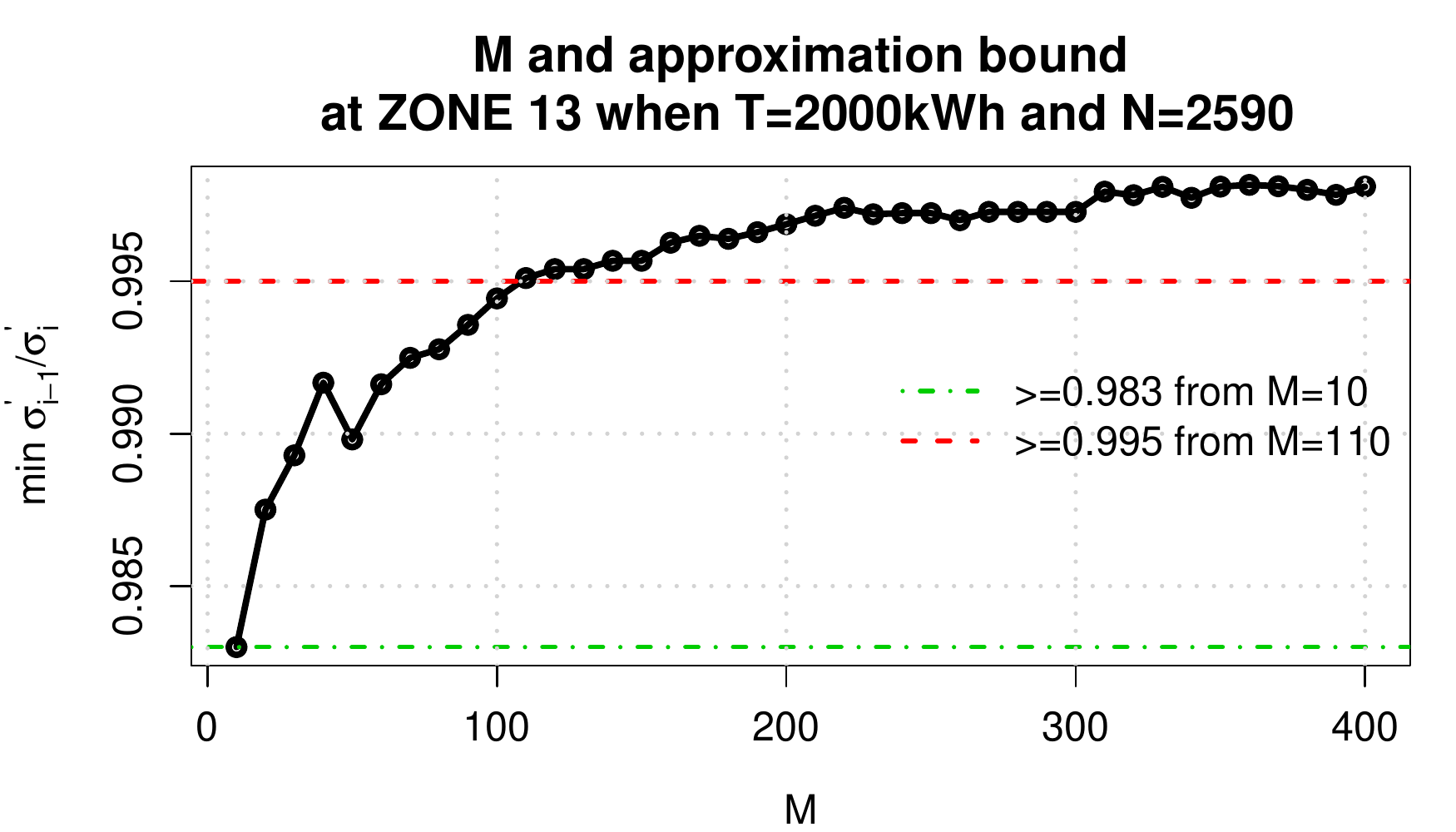}
\caption{M and approximation bound
}\label{mvsapprox}
\end{figure}
Finally, Fig. \ref{mvsapprox} 
shows the approximation bound 
of Algorithm \ref{alga2}. For $T$, we use 2000kWh, and 
for $N$, the minimum $N$ that achieves more than 95\% probability is selected for Zone 13. 
In Zone 13, from $M$=10, the approximation bound 
is 0.983, which means only 11 sortings of $K$ values are needed to achieve the almost optimal solution of \eqref{skp}. 

\section{Conclusions}\label{con}
In this paper, we investigated an efficient customer selection process for DR programs based on a large amount of hourly consumption data. We ensured that the proposed methodology scales to large data sets and is composed of a simple linear response modeling and a fast heuristic selection algorithm to solve the SKP. Moreover, 
we proved the approximation bound 
of the heuristic algorithm solution, which is an almost optimal solution in the hot climate zone. 


As an example DR program showing the result of our approach, we suggested a DR program controlling the temperature set point during the summer season. From the experimental results, we found that there are many customers, even in cool climate zones, who are very sensitive to the outside temperature and use their AC system actively. Therefore, when the number of recruiting customers, $N$, is small, the energy saving potential can be higher in the cool climate zone. The probability of achieving the targeted energy savings changes rapidly over a small range of $N$ in hot climate zones, which means it is very important to select the right customers to achieve the targeted energy savings with minimum enrollment and cost for the utility side.

The proposed method can be extended in many different ways. For example, the proposed heuristic algorithm does not have to be confined to a DR program targeting. If any application problem can be formulated in the same optimization problem format, the heuristic algorithm can be applied with low computational cost. Also, by changing DR programs and response modeling, more refined and practical research can be done. Additionally, it might be important to include other practical constraints or requirements from DR programs.

\appendices
\section{Gradual greedy approximation} \label{sec:appb}
We propose a very simple greedy algorithm to solve \eqref{skp}. Excluding the design parameter $T$ from \eqref{skp3}, we minimize $\boldsymbol{\mu}^T \bold{x}/\sqrt{{\boldsymbol{\sigma}^2}^T \bold{x}}$ with the assumption that $\bold{\Sigma}$ is diagonal. We sort customers by descending order of $\mu_k/\sigma_k$. 
However, when $T$ is larger than the sum of $\mu_k$ for the first $N$ customers, naive sorting can be far from the optimal solution. For example, when the $\mu_k$s are very small even with a high $\mu_k/\sigma_k$ ratio, it cannot make $\rho$ smaller than zero even though $\rho^*$ is smaller than zero. To prevent this drawback, we can modify the greedy approximation algorithm to the gradual greedy algorithm, Algorithm \ref{alga3}. This greedy algorithm will select one best customer at a time while satisfying the condition in \eqref{greedy}. Therefore, the algorithm guarantees that it achieves at least more than 50$\%$ when $\rho^* < 0$ in \eqref{skp3}. When $\rho^* > 0$, we let the greedy algorithm select the $N$ customers with the biggest $\mu_k$.
\begin{algorithm}
\caption{Gradual Greedy Algorithm}\label{alga3}
\begin{algorithmic}
\REQUIRE $\bold{\mu}$ and $\bold{\sigma^2}$ vectors in \eqref{indepcase}.
\STATE $\bold{x} = 0$; $T_0 = T$.
\FOR {$i$ from 1 to $N$}
\STATE Find $j$ from below and set $x_j$=1 in $\bold{x}$ :
\begin{align}\label{greedy}
j = \arg\max_{\{k| x_k = 0\}} &\frac{\mu_k}{\sigma_k},\nonumber\\
\textrm{(If $\rho^*\le0$)}\;s.t. \;&\mu_k \ge \frac{T_{i-1}}{N+1-i}.
\end{align}
\STATE \;$T_{i} = T_{i-1}-\mu_j$.
\ENDFOR
\end{algorithmic}
\end{algorithm}

\section{Proof for \textit{Proposition \ref{prop1}}}\label{approxproof}
\begin{figure}[htbp]
\includegraphics[width=3.5in,height=2in]{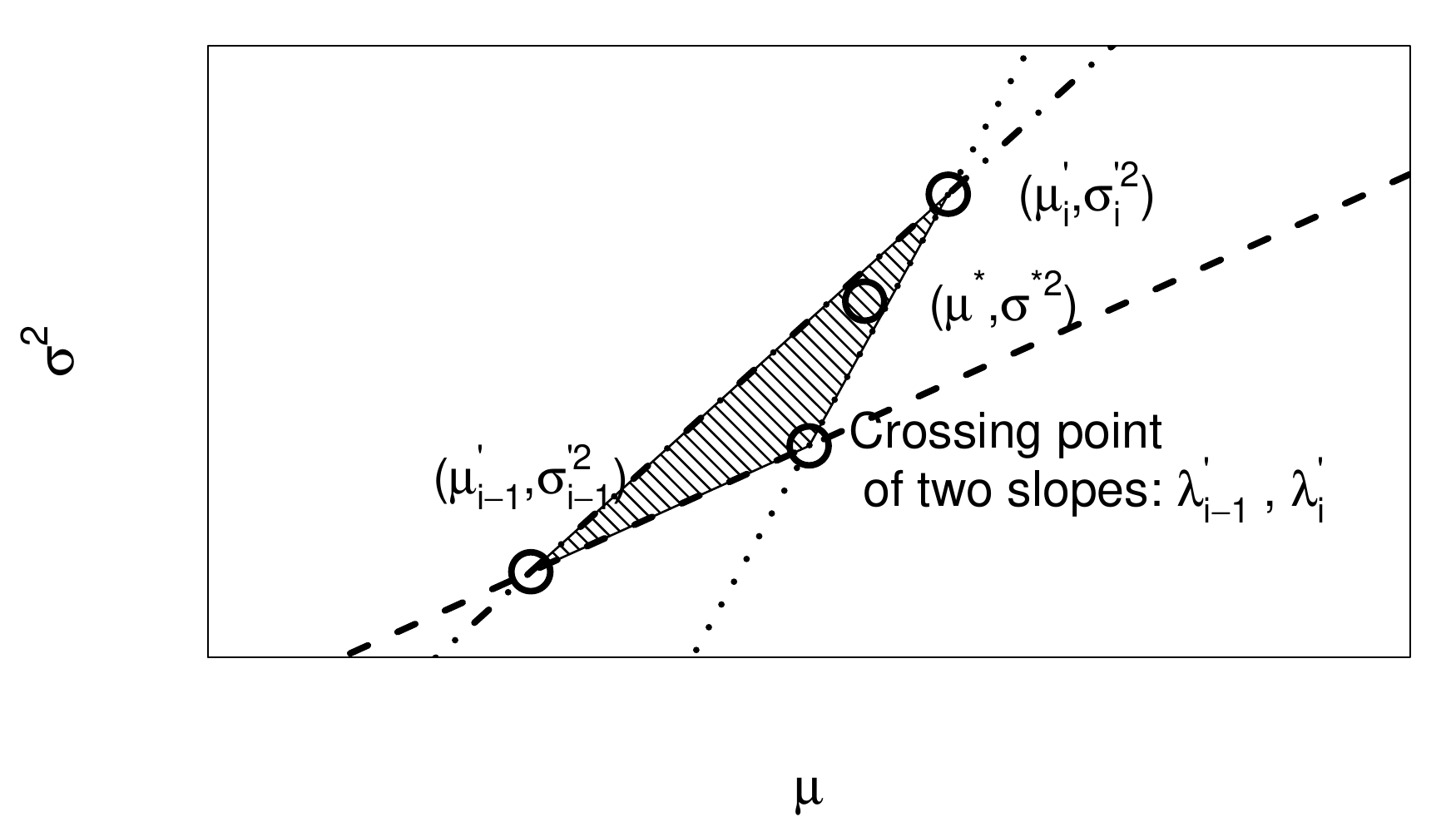}
\caption{When optimal point is not captured in Algoritham \ref{alga2}}\label{optimalpoint}
\end{figure}
\begin{proof}\label{proof2}
Let $\rho_{ALG}$ is the minimum $\rho$ obtained by Algorithm \ref{alga2}. And, we denote the optimal value of $\rho$ by $\rho^*$ in \eqref{skp3}. Briefly, if the optimal point $(\mu^*, \sigma^{*2})$ 
is one of the extreme points found by our heuristic algorithm (Algorithm \ref{alga2}), the algorithm achieves the actual optimal solution. However, as seen in Fig. \ref{optimalpoint}, the optimal point can be in the triangular region defined by the two consecutive slopes ($\lambda'_{i-1}, \lambda'_{i}$) and the two corresponding extreme points. Then, the optimal point satisfies the properties below.
\begin{align}
&\mu'_{i-1} < \mu^* < \mu'_i \textrm{ and } \sigma'^2_{i-1} < \sigma^{*2} < \sigma'^2_i\;,\; \exists i\\
&(\mu'_i = \boldsymbol{\mu}^T \bold{x_i},\; \sigma'^2_i = \bold{x_i}^T\bold{\Sigma x_i},\;\; \textrm{$\bold{x_i}$ is from (\ref{findextreme})}\;).\nonumber
\end{align}
Then, we can derive the inequality as seen in \eqref{approxproof}. 
\begin{align}\label{approxproof}
\rho_{ALG} &= \min_j \frac{T-\mu'_j}{\sigma'_j} \le \frac{T-\mu'_i}{\sigma'_i} < \frac{T-\mu^*}{\sigma'_i}\nonumber\\
&= \frac{T-\mu^*}{\sigma^*} \frac{\sigma^*}{\sigma'_i} < \rho^* \frac{\sigma'_{i-1}}{\sigma'_i}\le\rho^* \min_j \frac{\sigma'_{j-1}}{\sigma'_j}.
\end{align}
\end{proof}


From the proof, the approximation bound is given by the minimum ratio between two consecutive $\sigma^{'}_i$. It would be closer to 1 as $M$ is larger, which means Algorithm \ref{alga2} can achieve a solution 
within a very tight bound from the optimal solution. The relation between $M$ and the approximation bound 
on the actual data is provided in Fig. \ref{mvsapprox} 
in Section \ref{subsec:selection}.

\section*{Acknowledgment}
This research was funded in part by the Department of energy ARPA-E under award number DE-AR0000018, the California Energy Commission under award number PIR-10-054, and Precourt Energy Efficiency Center. The views and opinions of the authors expressed herein do not necessarily state or reflect those of the United States Government or any agency thereof.
 
We acknowledge the support of PG$\&$E for providing the anonymized data and in particular Brian Smith for his comments on this work. We also want to thank Sam Borgeson and Adrian Albert for their efforts to create the weather database.
\ifCLASSOPTIONcaptionsoff
  \newpage
\fi

    \bibliographystyle{ieeetr}
    \bibliography{papers}

\end{document}